\documentclass[sigplan,screen,nonacm,10pt]{acmart}

\AtBeginDocument{%
  }

\setcopyright{acmlicensed}
\copyrightyear{2018}
\acmYear{2024}
\acmDOI{XXXXXXX.XXXXXXX}

\acmConference[SPLASH 2024]{Proceedings of the 2023 ACM SIGPLAN International Symposium on New Ideas, New Paradigms, and Reflections on Programming and Software}{October 20--25, 2024}{Pasadena, California, USA}
\acmISBN{978-1-4503-XXXX-X/18/06}

\graphicspath{ {./images/} }
\usepackage{subcaption}
\usepackage{algorithm}
\usepackage{algpseudocode}
\usepackage{todonotes}
\usepackage{ulem}

\setcopyright{none}

\begin{document}

\title{Censor Resistant Instruction Independent Obfuscation for Multiple Programs}

\author{Ali Ajorian}
\affiliation{%
 \institution{University of Basel}
 \country{Switzerland}}
\email{ali.ajorian@unibas.ch}




\begin{abstract}
This work builds upon and optimizes our prior research on \textit{obfuscation as instruction decorrelation} which achieves multiple program obfuscation. Leveraging this infrastructure, we further achieve the property of sensor-resistant computation.
\end{abstract}



\received{25 April 2024}

\maketitle

\section {Introduction}

Accessing the binaries of a computer program raises an important question of how much can we learn about the program beyond simply running it with different inputs and observing its behavior. There are numerous techniques available that allow adversaries to analyze these binaries or run the program to study its behavior and internal states, with the intention of reverse engineering the source code or extracting secrets stored within the program. To protect software from such attacks, computer science has employed various ad-hoc heuristics for nearly four decades \cite{xu2020layered} and formal models for over two decades \cite{barak2001possibility, goldwasser2007best, kuzurin2007concept, ajorianiio}, collectively referred to as \textit{program obfuscation} or shortly \textit{obfuscation}. Roughly speaking, obfuscation is a transformation that maintains a program's functionality while hindering the extraction of valuable information, such as algorithms, data structures, and secret keys, with an acceptable slowdown overhead. 

Ajorian et. al introduced \textit{obfuscation as instruction decorrelation}, which focuses on the internal structures of regular programs to hide essential instructions within a set of junk instructions \cite{ajorianiio}. They defined the unintelligibility of an obfuscator as the decorrelation of all possible pairs of essential and junk instructions, ensuring that an efficient adversary cannot identify a correlated set of instructions to reconstruct an obfuscated program. 
 
In this work, we enhance the generation of junk instructions by obfuscating a collection of programs rather than focusing on a single program. Specially, for each program the remaining programs serve as the source of junk instructions. Since the obfuscator ensures that all obfuscated instructions are decorrelated, no efficient adversary can determine the origin of each instruction to disrupt the execution of a program. This means that an \textit{evaluator} executing the output of such an obfuscator can run either all or none of the input programs, making it impossible to censor any of the input programs.

\section{Preliminaries}
\label{sec:preliminaries}
This section introduces the notation that will be used consistently across this work and offers foundational background information relevant to our study.
 
\subsection{Notation}

PPT stands for probabilistic polynomial time algorithms, which can make random choices during execution, and their running time is bounded by a polynomial function of the input size.

By $x \leftarrow D$ we mean that $x$ is a random variable drawn from the distribution $D$. The notation $A(.) \rightarrow x$ denotes that a probabilistic algorithm A generates an outcome $x$ from an input $(.)$.

The notation $|.|$, in addition to its standard mathematical meanings such as the length of a set, is also utilized for indicating program size, referring to the number of elements or instructions within the program.

If $X$ is a set, by $X^*$ we mean all possible sequences of any length defined on the set X.

A function $\epsilon : \mathbb{N} \rightarrow \mathbb{N}$ is defined as negligible function if its growth rate is slower than the inverse of any polynomial. In other words, for any positive polynomial $p(n)$, there exists a threshold $n_0$ such that for all $n>n_0$, $\epsilon(n)<\frac{1}{p(n)}$.

By $set(p)$, we mean the set of all instructions in the program $p$.

\subsection{Background and Related Works}
\label{sec:background}

Let $\mathcal{P}$ be the set of all possible programs defined on an instruction space $\mathcal{I}$, Ajorian et. al defined an \textit{instruction-independent obfuscation} \cite{ajorianiio} as a PPT algorithm $IIO:\mathcal{P}\rightarrow \mathcal{P}$  with the following properties:

\begin{itemize}
	\item[--] \textbf{Functionality preservation}: $\hat p = IIO(p)$ computes the same functionality as input program $p$.
	\item[--] \textbf{Polynomial slowdown}: There exist polynomials $q$ and $r$ such that $|IIO(p)| \leq |p|$ and if $p$ halts within a maximum $t$ steps on input $x$, the obfuscated program halts within $r(t)$ steps on the same input.
	\item[--] \textbf{Unintelligibility}: for any auxiliary input $z$ of polynomial size, for any PPT adversary $A$ there must exist a negligible function $\epsilon$ such that:

	\begin{equation}
	\label{eq:s_unintelligibility}
	\begin{split}
		\big| & \Pr\nolimits_{A(\hat p, z)\rightarrow \hat s_1, \hat s_2; \hat s_1 \neq \hat s_2}[\{\hat s_1, \hat s_2\} \subseteq \mathbb{S} \text{ or } \{\hat s_1, \hat s_2\} \subseteq \mathbb{J}] - \\ 
      		& \Pr\nolimits_{\hat s_1, \hat s_2 \leftarrow \Omega ; \hat s_1 \neq \hat s_2}[\{\hat s_1, \hat s_2\} \subseteq \mathbb{S} \text{ or } \{\hat s_1, \hat s_2\} \subseteq \mathbb{J}] \big| \leq \epsilon(\lambda)
	\end{split}
	\end{equation}
\end{itemize}
where $\lambda$ is a user choice of a security parameter, $\mathbb{S}$ and $\mathbb{J}$ are the sets of essential and junk instructions respectively and $\Omega=set(\hat p)$.

\section{IIO with Multiple Programs}
An IIO obfuscator employs a set of junk instructions, denoted as $\mathbb{J}$, with specific requirements to conceal the set of essential instructions $\mathbb{S}$ of a source program $p$. Provided these requirements are satisfied, $\mathbb{J}$ can be constructed using various methods. Specially, $\mathbb{J}$ may consist of the union of several $\mathbb{S}_i$s where each $\mathbb{S}_i$ represents the essential instructions of a program $p_i$. From this perspective, an IIO obfuscator can be viewed as a transformation $\mathcal{O}$ that takes a set of programs $P=\{p_1, \ldots, p_n\}$ as inputs and produces a single obfuscated program $\mathcal{O}(P)$ that is semantically equivalent to the set $P$ while ensuring that an efficient adversary cannot correlate any pair of instructions originating from the same source program. 

The following definition encapsulates this concept of obfuscation:

\begin{definition}[Multiple Program IIO]
\label{RIObfs}
A PPT algorithm $\mathcal{O}$ is considered as a multiple program IIO for a set of input programs $P=\{p_1, \ldots, p_n\}$ if its output satisfies the following three conditions:
\begin{itemize}
	\item[--] \textbf{Functionality preservation}: $\mathcal{O}(P)$ computes the same functionality as running all individual programs in $P$.
	\item[--] \textbf{Polynomial slowdown}: There exist polynomials $q$ and $r$ such that $|\mathcal{O}(P)| \leq q(\sum_{i=1}^n |p_i|)$ and if $p_1, \ldots, p_n$ halts within a maximum $t$ of steps on inputs $x_1, \ldots, x_n$ respectively, the obfuscated program $\mathcal{O}(P)$ halts within $r(t)$ steps on the same input.
	\item[--] \textbf{Unintelligibility}: for any auxiliary input $z$ of polynomial size, for any PPT adversary $A$ and for all $i \in \{1, ..., n\}$ there exist a negligible function $\epsilon$ such that: 
	
	\begin{equation}
	\label{eq:mp_unintelligibility}
	\begin{split}
		\big| & \Pr\nolimits_{A(\mathcal{O}(P), z)\rightarrow \hat s_1, \hat s_2; \hat s_1 \neq \hat s_2}[\{\hat s_1, \hat s_2\} \subseteq \mathbb{S}_i ] - \\ 
      		& \Pr\nolimits_{\hat s_1, \hat s_2 \leftarrow \Omega ; \hat s_1 \neq \hat s_2}[\{\hat s_1, \hat s_2\} \subseteq \mathbb{S}_i ] \big| \leq \epsilon(\lambda)
	\end{split}
	\end{equation}
	where $\lambda$ is a user choice of a security parameter and $\mathbb{S}_i$ is the set of essential programs for the program $p_i$ and $\mathbb{J}_i =\bigcup_{j=1, j\neq i}^{n} \mathbb{S}_{j}$.
\end{itemize}
\end{definition}

Note that, what we have defined as unintelligibility in Equation  \ref{eq:mp_unintelligibility} is not directly derived from Equation \ref{eq:s_unintelligibility}. Rather, Equation \ref{eq:mp_bare_unintelligibility}, presented below, is a direct rewriting of Equation  \ref{eq:s_unintelligibility} for multi-program obfuscation. Furthermore, in Theorem \ref{thm:bare_unint_equivalence}, we prove that Equations \ref{eq:mp_unintelligibility} and \ref{eq:mp_bare_unintelligibility} are equivalent.

\begin{equation}
	\label{eq:mp_bare_unintelligibility}
	\begin{split}
		&\forall i \in \{1 \ldots n\}: \\
		\big| & \Pr\nolimits_{A(\mathcal{O}(P), z)\rightarrow \hat s_1, \hat s_2; \hat s_1 \neq \hat s_2}[\{\hat s_1, \hat s_2\} \subseteq \mathbb{S}_i \text{ or } \{\hat s_1, \hat s_2\} \subseteq \mathbb{J}_i] - \\ 
      		& \Pr\nolimits_{\hat s_1, \hat s_2 \leftarrow \Omega ; \hat s_1 \neq \hat s_2}[\{\hat s_1, \hat s_2\} \subseteq \mathbb{S}_i \text{ or } \{\hat s_1, \hat s_2\} \subseteq \mathbb{J}_i] \big| \leq \epsilon(\lambda)
	\end{split}
	\end{equation}

\begin{theorem}
\label{thm:bare_unint_equivalence}:
Equation \ref{eq:mp_unintelligibility} and Equation \ref{eq:mp_bare_unintelligibility} are equivalent.
\end{theorem}
\begin{proof}
\end{proof}

\subsection{Resistance to Censorship}
Equation \ref{eq:mp_unintelligibility} implies that no PPT adversary is able to correlate any two obfuscated instructions that originate from the same source program. This means that adversaries are unable to distinguish between input programs. 

In the context of delegating the execution of an IIO-obfuscated program to a computing agent, this property implies that the agent cannot identify or recognize the underlying program. As a result, the agent is unable to selectively halt or modify the computation based on the program's identity. In other words, the agent can either compute all of the programs or none of them, ensuring censor resistance of input programs.

\subsection{Achieving Verifiable Computation}
In scenarios where the computation of a function $f$ for an input $x$is delegated to a computing agent, a key challenge arises: how to ensure that the computing agent, also known as prover, has honestly computed the result using $f$. This problem, known as the \textit{Verifiable Computation Problem} (VCP), is of significant interest in the field of secure and trustworthy computation. The general solution to VPC,  involves having the prover provide a proof alongside the computed result. This proof enables the other party, referred to as the verifier, to efficiently confirm the correctness of the result without needing to recompute $f(x)$ as depicted in Figure \ref{fig:vcp}. 

\begin{figure}[th]
    \centering
    \includegraphics[width=0.3\textwidth]{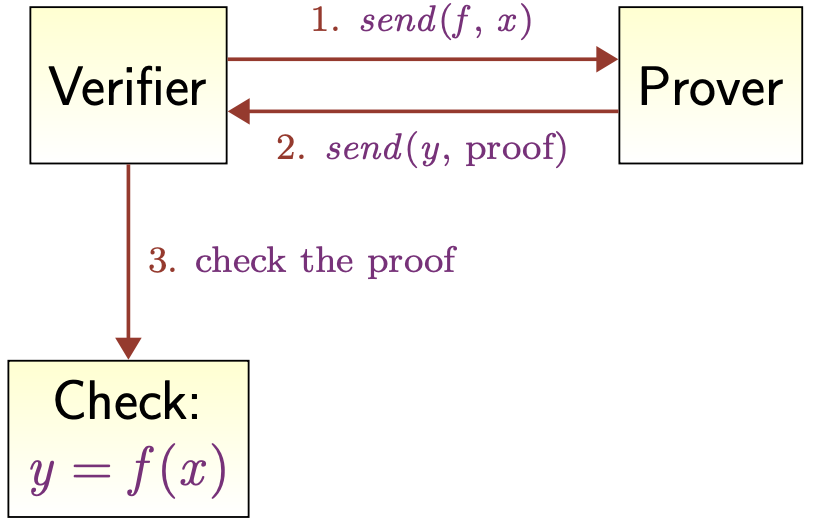}
    \caption{Verifiable Computing Problem}
    \label{fig:vcp}
\end{figure}

We can achieve verifiable computing using multiple program obfuscation. To do this, we need a light weight unknown program $c$ which serves as a proof in this framework. In our obfuscator, we consider the program set $P=\{f, c\}$. Since any IIO obfuscated program is censor-resistant, the prover cannot distinguish between $f$ and $c$ and $c$ remains unknown to the prover. During the process, the verifier provides inputs $x$ and $a$ as inputs for $f$ and $c$ respectively. The prover executes the obfuscated program and returns both $f(x)$ and $c(a)$ to the verifier. For verification, the verifier independently computes $c(a)$ and compares it with the value of $c(a)$ returned by the prover. If the values match, the verifier accepts the result $f(x)$. 

\section{Tampering Attack}
\label{sec:TampAttack}
Consider an obfuscated program $\mathcal{O}(P)$ that takes input vector $X=(x_1, \ldots, x_n)$ and produces output values $Y=(y_1, \ldots, y_n)$ for the programs $p_1, \ldots, p_n$ respectively. Throughout this process, the obfuscated program encounters a sequence of internal states $S=(s_1, \ldots, s_m)$. An adversary with the capability of tampering with the instructions or the internal states of $\mathcal{O}(P)$ can monitors $X,Y$ and $S$ before and after the tampering. The differences in these observations could potentially provide the adversary some information to correlate instructions and violate security guarantee of Equation \ref{eq:mp_unintelligibility}. Specially, the adversary could remove an instruction from the obfuscated program, execute the resulting program, and determine which output value $y_i$ changes. By doing this for polynomially many times, the adversary can detect correlations between the instructions of the obfuscated program.

We can achieve verifiable computation using IIO. A light weight function $c$ 



\subsection{Tamper-Resistance Output}
\label{label:Choutput}

To defend tampering attack, the obfuscator must guarantee that any attempt to tamper with $\mathcal{O}(P)$ either leaves the output vector $Y$ unchanged or or alters all its elements entirely. This can be accomplished through two approaches: (1) implementing a compensation mechanism that neutralizes the effects of tampered instructions, ensuring $Y$ remains unchanged, or (2) employing an encoding scheme for $Y$ that depends on all instructions of the obfuscated program, ensuring that any tampering affects the entire output vector uniformly. 

\section {Implementation}
\label{sec:construction}
In this section, we present an implementation of an obfuscator that hides which source program an obfuscated instruction originated from (Theorem~\ref{th:singleinst}) and  data access dependencies between instructions (Theorem~\ref{th:data-access-decorrelation}).

\subsection{Architecture Overview}
\label{sec:implementation:overview}


\subsection{Adversarial Model}

\section{Security Analysis}
\label{sec:sec-analysis}

\section{Implementation Results}
\label{sec:evaluation}

\section{Conclusions and Open Problems}
\label{sec:conclusion}

\begin{acks}
We would like to express our sincere gratitude to Osman Bicer for numerous insightful discussions concerning the security properties of this work.
\end{acks}

\bibliographystyle{ACM-Reference-Format}
\bibliography{main}

\appendix


\end{document}